\documentclass[11pt]{amsart}

\usepackage{amsmath}
\usepackage[foot]{amsaddr}
\usepackage{amssymb}
\usepackage{bm}
\usepackage{geometry}
\usepackage{amsthm}
\usepackage{amscd}
\usepackage{graphicx}
\usepackage{tikz-cd}
\usepackage{mathtools}
\usepackage[mathscr]{euscript}
\usepackage{setspace}
\usepackage{cjhebrew}
\usepackage{epigraph}
\usepackage{stmaryrd}
\usepackage{graphicx}
\usepackage{tikz}

\def\L8{L_\infty}

\def\fJ(E){\mathfrak E}
\def\fJ{\mathfrak J}

\usepackage{hyperref}
\hypersetup{
    colorlinks=true,
    linkcolor=blue,
    filecolor=red,      
    urlcolor=blue,
    filecolor=red,
    citecolor=blue,
}

\graphicspath{ {Desktop/} }
\newtheorem{thm}{Theorem}[section]
\newtheorem{conj}[thm]{Conjecture}

\newtheorem{prop}[thm]{Proposition}
\newtheorem{lem}[thm]{Lemma}
\newtheorem{cor}[thm]{Corollary}
\theoremstyle{definition}
\newtheorem{defn}[thm]{Definition}

\newtheorem{rmk}[thm]{Remark}
\newtheorem{const}[thm]{Construction}

\makeatletter
\newcommand*{\doublerightarrow}[2]{\mathrel{
  \settowidth{\@tempdima}{$\scriptstyle#1$}
  \settowidth{\@tempdimb}{$\scriptstyle#2$}
  \ifdim\@tempdimb>\@tempdima \@tempdima=\@tempdimb\fi
  \mathop{\vcenter{
    \offinterlineskip\ialign{\hbox to\dimexpr\@tempdima+1em{##}\cr
    \rightarrowfill\cr\noalign{\kern.5ex}
    \rightarrowfill\cr}}}\limits^{\!#1}_{\!#2}}}
\newcommand*{\triplerightarrow}[1]{\mathrel{
  \settowidth{\@tempdima}{$\scriptstyle#1$}
  \mathop{\vcenter{
    \offinterlineskip\ialign{\hbox to\dimexpr\@tempdima+1em{##}\cr
    \rightarrowfill\cr\noalign{\kern.5ex}
    \rightarrowfill\cr\noalign{\kern.5ex}
    \rightarrowfill\cr}}}\limits^{\!#1}}}
\makeatother

\title{Factorization Algebras for Linearized Gravity}

\author{Filip Dul}
\address{Department of Mathematics,
University of Hartford, West Hartford, CT 06117, USA}
\email{dul [AT] hartford [DOT] edu}

\begin{document}

\maketitle

\epigraph{I  promise nothing complete; because any human thing supposed to be complete, must for that very reason infallibly be faulty.}{Herman Melville, \emph{Moby-Dick}}

\begin{abstract}
	The purpose of this work is to bring gravitational theories into play within the quickly developing framework of factorization algebras. We fit the causal structure of Lorentzian manifolds into categorical language, and in the globally hyperbolic case discover a convenient equivalence of coverages. Then, we show how both perturbative general relativity and perturbative conformal gravity define Batalin-Vilkovisky classical field theories. Finally, we describe how the observables of linearized general relativity define a particularly nice factorization algebra on the category of Ricci flat globally hyperbolic manifolds and present a few conjectures which arise in specific cases, primarily motivated by the study of black hole entropy as a conserved Noether charge.
\end{abstract}

\tableofcontents

\section{Introduction}

This article is one in a set of papers related to joint work with Ryan Grady and others concerning factorization algebras arising in gravitation. Factorization algebras as defined in \cite{cosgwill1} and \cite{cosgwill2} are designed to minimally axiomatize the observables of a field theory as a (nearly) multiplicative cosheaf on spacetime. An active effort is underway to understand how factorization algebras are useful in understanding models from physics, like Yang-Mills theory and gravity. Work in the vein of \cite{bps}, \cite{gwillrejzfree}, and \cite{yau} form connections between the algebraic quantum field theory (AQFT) approach to functorial physics and factorization algebras adapted to the Lorentzian setting.  The current paper sets out on a program to apply constructions from the aforementioned sources to ascertain the usefulness of factorization algebras (and the Batalin-Vilkovisky formalism from which they arise) in the gravitational setting. The factorization algebra framework is undergoing rapid development, as described in \cite{cosgwill3}, and as that framework expands, we expect that interesting results and insights into gravitational theories will be discovered along the way. As such, we find it necessary to begin fusing into the Costello-Gwilliam approach to factorization algebras (as found in \cite{cosgwill2}) the traditional methods of gravitational theory as well as modern attempts to understand it via the Batalin-Vilkovisky methodology, as found in \cite{cattaneoschiavina} and elsewhere.

In Section \ref{section2}, we review the theory of causal structure of Lorentzian manifolds. Our approach will have two focal points: (1) we phrase much of our understanding in terms of the category $\mathbf{Loc}$ of all oriented and time-oriented globally hyperbolic Lorentzian manifolds, so as to define factorization algebras of observables on it, and (2) we keep an eye out on the structure of the moduli stack $[\mathrm{Lor}(M)/\mathrm{Diff}^{+}(M)]$ of all Lorentzian metrics modulo orientation and time-orientation preserving diffeomorphisms, which is quite different from its Riemannian analogue. Stacks in this instance provide a useful generalization of moduli spaces where stabilizers of the group action are ``remembered" in the quotient: this was our object of study in \cite{dul}. Although most of this section is review, Proposition (\ref{grothendieckequivalences}) is an original result which provides a convenient equivalence of coverages on $\mathbf{Loc}$ which fails in more general categories of spacetimes. A pioneering work which prioritized the value of working in the category $\mathbf{Loc}$ is \cite{brunetti}.

Section \ref{bvgravity} is the cornerstone of this paper. In it, we provide the most general prescription of a classical Batalin-Vilkovisky field theory to a generally covariant theory of gravity. General covariance admits a gauge freedom of the form $g \mapsto g + L_{X}g$, and we make this precise in Construction \ref{BVtheory}. In Section \ref{einstein}, we consider the specific case of general relativity--often called Einstein gravity--and show in Proposition \ref{einsteinisBV} that the formal moduli problem around a fixed Einstein metric $g$ in the stack $[\mathrm{Lor}(M)/\mathrm{Diff}^{+}(M)]$ of Lorentzian metrics modulo time-orientation preserving diffeomorphisms defines a classical BV theory. This is similar to a result in \cite{cattaneoschiavina}, but our version situates gravity as a BV theory in the stack-oriented formalization of \cite{cosgwill2}.  In Section \ref{weyl}, we discuss conformal gravity, which is sometimes called Weyl gravity.\footnote{The name is not ideal because of its similarity to ``Weyl \textit{geometry}", a separate field of study. So it goes.} This allows us to slightly enlarge the gauge freedom as $g \mapsto g + L_{X}g + 2fg$, for $f$ a smooth function on $M$. We show in Proposition \ref{weylisBV} that this defines a classical BV theory, and describe how in the case of a conformally flat and Einstein background metric $g \in [\mathrm{Lor}(M)/\mathrm{Diff}^{+}(M)]$, perturbative general relativity includes into perturbative conformal gravity as a BV theory. These are original results.

In Section \ref{factorizationalgebras}, we consider the underlying cochain complex of perturbative general relativity, which in the preceding section we viewed at the level of $L_{\infty}$ algebras. This truncation allows us to fit linearized general relativity into the context of \cite{bms}: we show in Theorem \ref{greenswitnessGR} that the associated cochain complex admits a Green's witness, which implies that the cochain complex is Green hyperbolic by results from \cite{bmsgreen}. As such, we show that the classical observables of linearized general relativity define a time-ordered and Cauchy constant prefactorization algebra on $\mathbf{RicLoc}$, the subcategory of Ricci flat globally hyperbolic spacetimes in $\mathbf{Loc}$: we explain how this follows from \cite{bms} in Corollary \ref{observablesforGR}. These results were partially inspired by the work in \cite{musante} and \cite{fewsterhunt}.

\section{Causal Structure}\label{section2}

\subsection{Lorentzian Geometry}
To begin our study of factorization algebras, we must recapitulate some essential ingredients of Lorentzian geometry and topology. We will focus mostly on the ideas introduced in Chapter 14 of \cite{oneil}. Other sources will be cited as needed. 

Recall that a Lorentzian metric $g$ on a manifold $M$ of dimension $n$ is a nondegenerate section of $\mathrm{Sym}^{2}T^{*}M$ with signature $(1,n-1)$, often denoted $(-,+,\ldots,+)$. In general relativity, we usually assume $M$ is four dimensional, so that the metric signature is $(1,3)$. A typical example of a Lorentzian manifold is Minkowski space $\mathbf{R}^{1,3}$, with metric denoted $\eta = -dt^{2} + dx^{2} + dy^{2} +dz^{2}$: at a given point, the norm squared of a tangent vector $v = \langle v_{0},v_{1},v_{2},v_{3} \rangle$ to $\mathbf{R}^{1,3}$ is therefore $\eta(v,v) = -v_{0}^{2} + v_{1}^{2} + v_{2}^{2} +v_{3}^{2}$. Every tangent space to a four dimensional Lorentzian manifold $M$ is modeled after $\mathbf{R}^{1,3}$, and so we arrive at the following definition. 

\begin{defn}\label{tangentvectors}
	(a) A tangent vector $v \in T_{p}M$ to a Lorentzian manifold $M$ is \textbf{timelike} if $g_{p}(v,v) < 0$, \textbf{spacelike} if $g_{p}(v,v) > 0$, and \textbf{null}, or \textbf{lightlike}, if $g_{p}(v,v) = 0$. A \textbf{causal} (or \textbf{non-spacelike}) vector is one which is either timelike or null. 
	(b) The \textbf{lightcone} (or \textbf{nullcone}) of $T_{p}M$ is the set of all null vectors based at $p$. 
\end{defn}

\begin{rmk}\label{noballs}
This definition lifts naturally to the case of a vector field $V$ on $M$: $V$ is called \textit{timelike} if $g(V,V) < 0$, etc. These definitions are what make the study of Lorentzian geometry so distinct. Although many constructions from Riemannian geometry carry over nicely, others do not. Significantly, \textit{there is no good embedding of the category of Lorentzian manifolds into the category of metric spaces} \cite{mullersanchez}. This means that there is no analogue of the Hopf-Rinow Theorem in the Lorentzian setting, and the topology generated by small open balls is not a useful one, since a ball of any radius is not bounded. 
\end{rmk}

The lightcone in $T_{p}M$ splits into two parts: a future lightcone and a past lightcone. The future lightcone is called $p$'s \textit{future}, and the past lightcone is called $p$'s \textit{past}.  We are mostly interested in Lorentzian manifolds which have a consistent choice of future and past throughout $M$: i.e., they are ``time-orientable". The following definition is adapted from Lemma 5.32 in \cite{oneil}.

\begin{defn}
	A Lorentzian manifold is \textbf{time-orientable} if and only if there exists an everywhere timelike continuous vector field $V$. The equivalence class of $V$ under the relation that the timelike vectors remain in the same half of the lightcone is a \textbf{time orientation} on $M$. Such a time-oriented Lorentzian manifold is called a \textbf{spacetime}.
\end{defn}

\begin{rmk}
In light of this, the moduli stack for general relativity is $[\mathrm{Lor}(M)/\mathrm{Diff}^{+}(M)]$, the stack of time-oriented Lorentzian metrics on $M$ modulo orientation and time-orientation preserving diffeomorphisms. This is a particularly interesting stack because $\mathrm{Lor}(M)$ itself is not a convex cone in the way that $\mathrm{Riem}(M)$ is, as described in \cite{mullersanchez}.  $\mathrm{Lor}(M)$ usually has multiple connected components and is therefore also not contractible. 
\end{rmk}

To globalize the local information of Definition \ref{tangentvectors}, we introduce the following:

\begin{defn}\label{causality}
	(a) A curve from a point (or \textbf{event}) $p$ to a point $q$ in $M$ is called \textbf{timelike} if each of its tangent vectors is timelike.\footnote{The analogous definition holds for spacelike, null, and causal.} 
	(b) For events $p,q \in M$, we say $p << q$ if there exists a future-pointing timelike curve in $M$ from $p$ to $q$, and we say $p<q$ if there exists a future-pointing causal curve in $M$ from $p$ to $q$. (c) The \textbf{chronological future of $p$} is the set $I^{+}(p) := \{ q \in M : p <<q \}$ and the \textbf{causal future of $p$} is the set $J^{+}(p) := \{ q \in M : p \leq q \}$.\footnote{By ``$p \leq q$", we are allowing $p=q$. This means that $A \cup I^{+}(A) \subset J^{+}(A)$ in the following remark.} The \textbf{chronological past} and \textbf{causal past} of $p$, denoted $I^{-}(p)$ and $J^{-}(p)$, are defined analogously. (d) Given a subset $A$ of $M$, its \textbf{chronological future} is $I^{+}(A)$: the union of all $I^{+}(p)$ such that $p \in A$. Similar definitions hold for $I^{-}(A)$ and $J^{+/-}(A)$.
\end{defn}

One could put $M$ in the subscript--for example $I^{+}_{M}(A)$ instead of $I^{+}(A)$--to be explicit, which we occasionally do. There is an important way of distinguishing sets which are meaningfully spacelike,  in the sense that a timelike or causal curve only goes through it at one point.

\begin{defn}
	A subset $A$ of $M$ is \textbf{acausal} if the relation $p << q$ does not hold for any $p,q \in A$. A subset $B$ is \textbf{achronal} if the relation $p < q$ does not hold for any $p,q \in B$.
\end{defn}

Note that acausal implies achronal, but not vice versa. A typical example of achronal sets are the constant time hypersurfaces $\{t = t_{0}\}$ in $\mathbf{R}^{1,3}$. These hypersurfaces are also examples of \textit{Cauchy surfaces}, which are best defined in the context of \textit{Cauchy developments}. 

\begin{defn}\label{cauchy}
	(a) The \textbf{future Cauchy development} of a set $A \subset M$ is the set $D^{+}(A)$ of all points $p$ in $M$ such that every past-inextendible causal curve through $p$ meets $A$. The \textbf{past Cauchy development} $D^{-}(A)$ is defined analogously, and the \textbf{Cauchy development} of $A$ is $D(A) := D^{+}(A) \cup D^{-}(A)$. (b) A \textbf{Cauchy surface} (or  Cauchy hypersurface) is an achronal set $C$ such that $D(C) = M$. 
\end{defn}

	Equivalently, a Cauchy surface is a set $C$ which is met exactly once by every inextendible timelike curve in $M$. One may heuristically imagine that $C$ represents a ``snapshot in time" of the spacetime $M$. In fact, when $M$ is \textit{globally hyperbolic}, it can be decomposed as $\mathbf{R} \times C$.
	
\begin{defn}
A spacetime $M$ is \textbf{globally hyperbolic} if it is causal (i.e. contains no closed causal curves) and if for all $p,q \in M$, the sets $J^{+}(p) \cap J^{-}(q)$ are compact.
\end{defn}

\begin{rmk}
	If $p << q$, then $I^{+}(p) \cap I^{-}(q)$--often called a \textit{causal diamond}--is itself a globally hyperbolic spacetime, whose Cauchy surface is a subset of $C$ (as long as it is not empty). More often, one would think of $I^{+}(p) \cap I^{-}(q)$ as a ``subspacetime" since it is an open set. 
	
	The sets $I^{+}(p) \cap I^{-}(q)$ for all $p,q \in M$ are sometimes called \textit{Alexandrov intervals}, and generate what is known as the \textit{Alexandrov topology}. The Alexandrov topology coincides with the ordinary manifold topology of $M$ if and only if $M$ is strongly causal (i.e. $M$ contains no causal curves which are \textit{almost} closed, in a way which can be made precise) \cite{minguzzi}. If a spacetime is globally hyperbolic, this is the case; in this way, one can think of the Alexandrov topology as an analogue of the topology generated by open metric balls in a Riemannian manifold, which we mentioned in Remark \ref{noballs} does not port over to the Lorentzian case.
\end{rmk}

In the home stretch to our definition of factorization algebras for the category $\mathbf{Loc}$ of globally hyperbolic spacetimes, we must slightly generalize the idea of causal diamonds.

\begin{defn}
	(a) A subset $S \subseteq M$ of a globally hyperbolic time-oriented Lorentzian manifold is called \textbf{causally convex} if $J^{+}(S) \cap J^{-}(S) \subseteq S$. (b) Two subsets $S, S'$ are called \textbf{causally disjoint} if $(J^{+}(S) \cup J^{-}(S)) \cap S' = \varnothing$. 
\end{defn}

In particular, $S \subseteq M$ is causally convex if every causal curve which begins and ends in $S$ is contained entirely in $S$. With this, we can define our main category of interest.

\begin{defn}
	$\mathbf{Loc}$ is defined to be the category whose objects are all oriented and time-oriented globally hyperbolic Lorentzian manifolds (i.e. all oriented globally hyperbolic spacetimes) of a fixed dimension and whose morphisms are the orientation and time-orientation preserving isometric embeddings $f: M \to N$ with causally convex and open image $f(M) \subseteq N$. 
\end{defn}

\begin{defn}
	(a) A $\mathbf{Loc}$-morphism $f : M \to N$ is called a \textbf{Cauchy morphism} if its image $f(M) \subseteq N$ contains a Cauchy surface of $N$. These are often denoted $f : M \xrightarrow{c} N$. \\
	(b) A pair $(f_{1} : M_{1} \to N, f_{2} : M_{2} \to N)$ of $\mathbf{Loc}$-morphisms is called \textbf{causally disjoint} if their images $f_{1}(M_{1})$ and $f_{2}(M_{2})$ are causally disjoint in $N$: denoted $f_{1} \perp f_{2}$. \\
	(c) An $n$-tuple of $\mathbf{Loc}$-morphisms $(f_{1} : M_{1} \to N, \ldots, f_{n} : M_{n} \to N)$ is called \textbf{pairwise disjoint} if the images $f_{i}(M_{i})$ are pairwise disjoint in $N$: denoted $\underline{f} = (f_{1},\ldots,f_{n}) : \underline{M} \to N$. 
\end{defn}

\subsection{Factorization Algebras} Given a classical field theory coming from an action functional on its field space $\mathscr{E}$ and with underlying manifold/spacetime $M$, we extract from that functional the equations of motion of the theory. Because equations of motion are PDEs, which impose constraints locally on $M$, they define\footnote{At least in the case of smooth or distributional fields, or others which have constraints imposed \textit{locally}.} a sheaf $\mathrm{Sol}_{\mathscr{E}} : \mathrm{Open}(M) \to \mathscr{C}$, where $\mathscr{C}$ denotes a symmetric monoidal category. Depending on the context, $\mathscr{C}$ is defined so that $\mathrm{Sol}_{\mathscr{E}}(U)$ is a set, vector space, or cochain complex of solutions restricted to $U$. 

The observables $\mathrm{Obs}_{\mathscr{E}}$ are the function ring on $\mathrm{Sol}_{\mathscr{E}}$, and under some assumptions form a cosheaf on $M$. Moreover, for disjoint $U,V \in \mathrm{Open}(M)$, the equivalence $\mathrm{Sol}_{\mathscr{E}}(U \sqcup V) \cong \mathrm{Sol}_{\mathscr{E}}(U) \times \mathrm{Sol}_{\mathscr{E}}(V)$ gives a structure map $\mathrm{Obs}_{\mathscr{E}}(U) \otimes \mathrm{Obs}_{\mathscr{E}}(V) \to \mathrm{Obs}_{\mathscr{E}}(U \sqcup V)$, making $\mathrm{Obs}_{\mathscr{E}}$ a multiplicative cosheaf called a \textit{factorization algebra}. Precise details on this construction are provided in \cite{cosgwill1} and elsewhere, but we are particularly interested in a variant of the formal definition which is provided in \cite{bps}. We will restrict the source category to open sets which are causally well-behaved, and $\mathscr{C} = \mathbf{Ch}$ will denote the category of cochain complexes unless otherwise stated.

\begin{defn}\label{PFA}
	 A \textbf{prefactorization algebra} $\mathcal{F}$ on $\mathbf{Loc}$ with values in $\mathbf{Ch}$ is a functor assigning: (i) for each $M \in \mathbf{Loc}$, an object $\mathcal{F}(M) \in \mathbf{Ch}$ and (ii) for each tuple $\underline{f} = (f_{1},\ldots,f_{n}) : \underline{M} \to N$ of pairwise disjoint morphisms a $\mathbf{Ch}$-morphism $\mathcal{F}(\underline{f}) : \bigotimes_{i=1}^{n} \mathcal{F}(M_{i}) \to \mathcal{F}(N)$--called a \textbf{factorization product}--with the convention that $\varnothing \to N$ is assigned a morphism $I \to \mathcal{F}(N)$ from the monoidal unit. 
	This data must satisfy: (1) for every $\underline{f} = (f_{1},\ldots,f_{n}) : \underline{M} \to N$ and $\underline{g}_{i} = (g_{i1}, \ldots, g_{ik_{i}}) : \underline{L}_{i} \to M_{i}$, for $i = 1, \ldots, n$, the diagram
	\[\begin{tikzcd}
	{\bigotimes_{i=1}^{n} \bigotimes_{j=1}^{k_{i}} \mathcal{F}(L_{ij})} && {\bigotimes_{i=1}^{n} \mathcal{F}(M_{i})} \\
	\\
	&& {\mathcal{F}(N)}
	\arrow["{\otimes_{i} \mathcal{F}(\underline{g}_{i})}", from=1-1, to=1-3]
	\arrow["{\mathcal{F}(\underline{f}(\underline{g}_{i},\ldots,\underline{g}_{n})) }"', from=1-1, to=3-3]
	\arrow["{\mathcal{F}(\underline{f})}", from=1-3, to=3-3]
\end{tikzcd}\]
commutes, where $\underline{f}(\underline{g}_{i},\ldots,\underline{g}_{n}) := (f_{1}g_{11},\ldots, f_{n}g_{nk_{n}})$ is given by composition in $\mathbf{Loc}$; (2) for every $M \in \mathbf{Loc}$, $\mathcal{F}(\mathrm{id}_{M}) = \mathrm{id}_{\mathcal{F}(M)}$; and (3) for every $\underline{f} = (f_{1},\ldots,f_{n}) : \underline{M} \to N$ and permutation $\sigma \in S_{n}$, the diagram \[\begin{tikzcd}
	{\bigotimes_{i=1}^{n} \mathcal{F}(M_{i})} && {\mathcal{F}(N)} \\
	\\
	{\bigotimes_{i=1}^{n} \mathcal{F}(M_{\sigma(i)})}
	\arrow["{\mathcal{F}(\underline{f})}", from=1-1, to=1-3]
	\arrow["\sigma"', from=1-1, to=3-1]
	\arrow["{\mathcal{F}(\underline{f}\sigma)}"', from=3-1, to=1-3]
\end{tikzcd}\] commutes, where $\underline{f}\sigma$ is given by right permutation. A morphism $\zeta : \mathcal{F} \to \mathcal{G}$ of prefactorization algebras is a family $\zeta_{M} : \mathcal{F}(M) \to \mathcal{G}(M)$ of $\mathbf{Ch}$-morphisms for all $M \in \mathbf{Lor}$ (i.e. a natural transformation) which is compatible with factorization products. The category of all prefactorization algebras on $\mathbf{Loc}$ will be denoted $\mathbf{PFA}$.
\end{defn}

\begin{defn}\label{timeorderedPFA}[Definition 4.1 of \cite{bps}]
	An $n$-tuple $\underline{f} = (f_{1},\ldots,f_{n}) : \underline{M} \to N$ is called \textbf{time-ordered} if $J^{+}_{N}(f_{i}(M_{i})) \cap f_{j}(M_{j}) = \varnothing $ for $i < j$. Such an $n$-tuple is called \textbf{time-orderable} if there exists a permutation of the indices--called a \textbf{time-ordering permutation}--which makes it time-ordered. A \textbf{time-orderable prefactorization algebra} takes as an input a time-ordered $n$-tuple $\underline{f} = (f_{1},\ldots,f_{n}) : \underline{M} \to N$ and outputs a \textbf{time-ordered product} $\mathcal{F}(\underline{f}) : \bigotimes_{i=1}^{n} \mathcal{F}(M_{i}) \to \mathcal{F}(N)$.\footnote{Otherwise, the rest of this definition follows analogously from that of Definition \ref{PFA}.}
\end{defn}

\begin{rmk}
	A common definition of prefactorization algebra focuses on a fixed manifold or spacetime $M$, and $\mathbf{Loc}$ is replaced by the category of open sets (or in our case the category of open sets which are also elements of $\mathbf{Loc}$) on $M$. We will now present a category which combines the two perspectives: we consider opens in a fixed $M$ which are in $\mathbf{Loc}$. 
\end{rmk}

\begin{defn}
	For $M \in \mathbf{Loc}$, let $\mathbf{RC}_{M}$ denote the category of relatively compact and causally convex open sets in $M$ with morphisms given by inclusion. 
\end{defn}

	For every $M \in \mathbf{Loc}$, $\mathbf{RC}_{M}$ is a directed set \cite{bps}. The orientation, time-orientation, and the Lorentzian metric from $M$ restrict to $U \in \mathbf{RC}_{M}$ so that $U \in \mathbf{Loc}$. The $\mathbf{RC}_{M}$-morphism $U \subseteq V$ defines a $\mathbf{Loc}$-morphism $\iota_{U}^{V} : U \to V$. Hence, we see that $\mathbf{RC}_{M} \subset \mathbf{Loc}$ is a subcategory, and so we can restrict any $\mathcal{F} \in \mathbf{PFA}$ to a functor $\mathcal{F}|_{M} : \mathbf{RC}_{M} \to \mathbf{Ch}$. This is the traditional notion of a prefactorization algebra from \cite{cosgwill1} in that it is defined on a fixed manifold $M$ and the source category is a (sub)category of open sets on $M$.
	
	For globally hyperbolic spacetimes, the Alexandrov intervals $I^{+}(p) \cap I^{-}(q)$ are the ``typical" elements of $\mathbf{RC}_{M}$.\footnote{For a non-globally hyperbolic spacetime, they may not be relatively compact!} We can make this precise as follows. Let $\mathcal{A}_{M}$ denote the category of Alexandrov intervals for $M \in \mathbf{Loc}$. For each $M \in \mathbf{Loc}$, we have a covering family $\{ U_{i} \to M \}$ where each $U_{i}$ is in $\mathcal{A}_{M}$. If $f : N \to M$ is a $\mathbf{Loc}$-morphism, then $\{ f^{-1}(U_{i}) \to N \}$ may not all be relatively compact, but a relatively compact refinement exists. Moreover, the sets $f^{-1}(U_{i})$ may not be Alexandrov intervals, but they are causally convex. As such, for each point $x_{i} \in f^{-1}(U_{i})$, an Alexandrov interval $I^{+}(p_{i}) \cap I^{-}(q_{i})$ may be found such that $I^{+}(p_{i}) \cap I^{-}(q_{i}) \subset  f^{-1}(U_{i})$,\footnote{The construction is a bit tedious, so we omit it here.} so that an Alexandrov refinement of $\{f^{-1}(U_{i}) \}$ exists. As such, we can find a common refinement  of the two.
	
	In this way, $\mathcal{A}$ defines a coverage on $\mathbf{Loc}$. $\mathbf{RC}$ similarly defines a coverage on $\mathbf{Loc}$, and indeed, for a fixed $M \in \mathbf{Loc}$, the above argument shows that these coverages define equivalent topologies on $M$. This can be slightly adapted to arrive at the following convenient fact.

\begin{prop}\label{grothendieckequivalences}
$\mathcal{A}$ and $\mathbf{RC}$ define equivalent coverages on the site $\mathbf{Loc}$. 
\end{prop}

\begin{rmk}
	To boot, as shown in \cite{bps}, $\mathbf{RC}_{M}$ defines a \textit{Weiss cover} for $M$: namely, given any finite collection $\{x_{1},\ldots,x_{k}\}$ in $M$, there exists $U \in \mathbf{RC}_{M}$ not equal to all of $M$ which contains these points. A prefactorization algebra is a \textbf{factorization algebra} if it is a cosheaf with respect to every Weiss cover of $M$. Next, we provide an intermediary quality between a prefactorization algebra and a factorization algebra.
\end{rmk}

\begin{defn}[2.11 in \cite{bps}]
	A prefactorization algebra $\mathcal{F} \in \mathbf{PFA}$ is called \textbf{additive} if for every $M \in \mathbf{Loc}$, the canonical morphism $\mathrm{colim}(\mathcal{F}|_{M} : \mathbf{RC}_{M} \to \mathbf{Ch}) \to \mathcal{F}(M)$ is a $\mathbf{Ch}$-isomorphism. $\mathbf{PFA}^{\mathrm{add}}$ will denote the full subcategory of additive prefactorization algebras.
\end{defn}

The idea behind this definition is that the images of the maps $\mathcal{F}(U) \to \mathcal{F}(M)$ ``generate" $\mathcal{F}(M)$ for all $U \in \mathbf{RC}_{M}$. This is not quite as strong as being a factorization algebra, but is useful nonetheless. Moreover, since $\mathbf{RC}_{M}$ is a Weiss cover, we get the following.

\begin{prop}[2.13 in \cite{bps}]
	Every factorization algebra $\mathcal{F}$ on $\mathbf{Loc}$ is an additive prefactorization algebra. 
\end{prop}

\begin{defn}
	A time-orderable prefactorization algebra $\mathcal{F}$ on $\mathbf{Loc}$ is called \textbf{Cauchy constant} if it satisfies the \textbf{time-slice axiom}: namely, for every Cauchy morphism $f : M \xrightarrow{c} N$, $\mathcal{F}(M) \to \mathcal{F}(N)$ is a quasi-isomorphism. 
\end{defn}

\begin{rmk}
Cauchy constancy is meant to be the analogue of a locally constant prefactorization algebra in the sense of \cite{cosgwill1}, but the analogy cannot be taken too far, because the latter are usually associated to topological field theories. We do not expect Cauchy constant prefactorization algebras to come from topological field theories: indeed, their very definition presupposes the relevance of a Lorentzian metric. Certain physically viable field theories on spacetimes should satisfy the Cauchy constancy axiom. In the next section, we will introduce linearized gravity as an example of a Batalin-Vilkovisky (BV) theory to set up the study of its factorization algebras of observables and see what properties it has.
\end{rmk}

\section{Perturbative Gravity as a BV Theory}\label{bvgravity}

We will refrain from recapping all of the BV formalism: Section 2.1 of \cite{dul} provides a rapid introduction, and we may cite results from there and \cite{cosgwill2} as needed. But we will specify the definition we use to avoid any ambiguities. $\mathscr{E}$ denotes the smooth sections of $E \to M$, and a precise definition of local functionals $\mathscr{O}_{\mathrm{loc}}(\mathscr{E})$ can be found in Section 2 of \cite{dul}, although one should view them as polynomial functions of $\mathscr{E}$ which are of integral form.

\begin{defn}[Definition 2.10 in \cite{dul}]\label{BVdefn}
A \textbf{Batalin-Vilkovisky classical field theory} $(\mathscr{E}, \omega, S_{BV})$ on a smooth manifold $M$ is a differential $\mathbf{Z}$-graded vector bundle $E \to M$ equipped with a $-1$-shifted symplectic structure $\omega$ and an action functional $S_{BV} \in \mathscr{O}_{\mathrm{loc}}(\mathscr{E})$ such that: \\
(1) $S_{BV}$ satisfies the \textbf{classical master equation} (CME): $\{S_{BV},S_{BV}\} = 0$, and \\
(2) $S_{BV}$ is at least quadratic, so that it can be written uniquely as $S_{BV}(\varphi) = \omega(\varphi, Q\varphi) + I(\varphi)$, where $Q$ is the linear differential operator of the above DG structure and $I \in \mathscr{O}_{\mathrm{loc}}(\mathscr{E})$ contains only third or higher degree terms.  \\
A \textbf{free theory} is one in which $I=0$: i.e. the action functional $S_{BV}$ is purely quadratic.
\end{defn}

To begin, we must consider the relevant global moduli stack and extract from it the formal moduli problems for perturbative gravity around fixed solutions of the nonlinear equations. For a fixed spacetime $M$, the moduli stack of a generally covariant (i.e. diffeomorphism equivariant) metric theory of gravity is $[\mathrm{Lor}(M)/\mathrm{Diff}^{+}(M)]$, the stack of time-oriented Lorentzian metrics on $M$ modulo time-orientation and orientation preserving diffeomorphisms. If we fix a point $g \in [\mathrm{Lor}(M)/\mathrm{Diff}^{+}(M)]$, then in accordance with Section 3 of \cite{dul}, its shifted tangent complex is 
\begin{equation}\label{dgla}
	\mathfrak{K}_{g} := 0 \to \mathrm{Vect}_{M} \xrightarrow{L_{\bullet}g} \Gamma (\mathrm{Sym}^{2}T_{M}^{*})[-1] \to 0.
\end{equation}	
$\mathfrak{K}_{g}$ is an example of a differential graded Lie Algebra (or DGLA).\footnote{A precise definition is in Appendix A of \cite{cosgwill2}. Note that a DGLA is an example of an $L_{\infty}$ algebra, but where only the $\ell_{1}$ and $\ell_{2}$ terms are nontrivial.} This is often referred to as the DGLA of Killing fields, since it cohomologically resolves Killing fields.

\begin{rmk}
	The differential $\ell_{1}$ in this case is $L_{\bullet}g$, and the Lie bracket $\ell_{2} = [-,-]$ is nontrivial because the vector fields are situated in degree 0. If we did not shift the tangent complex up by one, they would be in degree $-1$, which means $[X,Y]$ would have to land in degree $-2$, which is 0 in this case. Because $[X,-] \in \mathrm{End}(\Gamma(\mathrm{Sym}^{2}T_{M}^{*}))$ is defined to be $L_{X}(-)$, we get that for $X, Y \in \mathrm{Vect}_{M}$,
	\begin{equation}\label{dglacomputation}
		L_{[X,Y]}g + [L_{X}g, Y] - [L_{Y}g,X] = L_{[X,Y]}g - L_{Y}L_{X}g + L_{X}L_{Y}g = 0.
	\end{equation}
	This shows us that $L_{\bullet}g$ is indeed a graded derivation of the Lie bracket $[-,-]$ on vector fields, and so confirms that $\mathfrak{K}_{g}$ is indeed a DGLA. This same computation, but for arbitrary $h \in \Gamma(\mathrm{Sym}^{2}T_{M}^{*})$ instead of the fixed $g$, shows that $\mathfrak{K}_{g}$ is an extension of $\mathrm{Vect}_{M}$ by a $\mathrm{Vect}_{M}$-module.
\end{rmk}

\begin{const}\label{BVtheory}
Fix a solution $g \in [\mathrm{Lor}(M)/\mathrm{Diff}^{+}(M)]$ to the nonlinear equations of a given vacuum gravitational theory. According to Section 4 of \cite{cosgwill2}, the formal moduli problem around this point is equivalent to an $L_{\infty}$ algebra $\mathscr{E}_{g}$, which we now construct. The underlying cochain complex of $\mathscr{E}_{g}$ is the shifted cotangent bundle $T^{\vee}[-1]\mathfrak{K}_{g}$ of (\ref{dgla}) which is then ``twisted" by the linearized differential operator, which we denote by $D\mathbf{E}_{g}$: this is the differential operator defined by linearizing the \textit{nonlinear} differential operator $\mathbf{E}$ of the original theory around a fixed background solution $g$ of that nonlinear operator. 

We denote the underlying cochain complex, or DG fields, as
\begin{equation}\label{BVfields}
	(T^{\vee}[-1]\mathfrak{K}_{g}, D\mathbf{E}_{g}) = \mathrm{Vect}_{M} \xrightarrow{L_{\bullet}g} \Gamma(\mathrm{Sym}^{2}T_{M}^{*})[-1] \xrightarrow{D\mathbf{E}_{g}} \Gamma(\mathrm{Sym}^{2}T_{M})[-2] \xrightarrow{L_{\bullet}g^{*}} {\Omega^{1}_{M}}[-3],
\end{equation}
or even simply $\mathscr{E}_{g}$ if it is clear that we are referring to the underlying DG fields of the $L_{\infty}$ algebra. The degree 2 part $\Gamma(\mathrm{Sym}^{2}T_{M})$ and the degree 3 part ${\Omega^{1}_{M}}$ are chosen so that they pair with the degree 1 and degree $0$ parts, respectively. Given $h_{\mu\nu} \in \Gamma(\mathrm{Sym}^{2}T_{M}^{*})$ and $(h^{\dag})^{\mu\nu} \in \Gamma(\mathrm{Sym}^{2}T_{M})$ with compact overlapping support, the shifted symplectic pairing is: 
\begin{equation}\label{shiftedsymplectic}
 \langle h, h^{\dag} \rangle := \int_{M} h_{\mu\nu}(h^{\dag})^{\mu\nu} \mathrm{vol}_{g},
\end{equation}
so that the action functional for the linearized theory is 
\begin{equation}
	\langle h, D\mathbf{E}_{g}(h) \rangle = \int_{M} h_{\mu\nu}(D\mathbf{E}_{g}(h))^{\mu\nu} \mathrm{vol}_{g}.
\end{equation}
The shifted symplectic pairing is defined analogously for the ghosts $X$ and antighosts $X^{\dag}$.

The final term in the differential $\ell_{1}$, which we denote generically by $L_{\bullet}g^{*}$, is natural: it is the $L^{2}$ dual of the Lie derivative, which is the natural divergence operator $\mathrm{div}_{g} : \Gamma(\mathrm{Sym}^{2}T^{*}_{M}) \to \mathrm{Vect}_{M}$. In coordinates, it is
$
(L_{\bullet}g^{*})(h) := -2\mathrm{div}_{g}(h) = -2\nabla^{\mu}{h_{\mu}}^{\nu}.
$
For $h^{\dag} \in \Gamma(\mathrm{Sym}^{2}T_{M})$, and to have it land in $\Omega^{1}_{M}$ as we would like, our version has to be 
\begin{equation}
	(L_{\bullet}g^{*})(h^{\dag}) := -2\mathrm{div}_{g}(h^{\dag}) = -2\nabla_{\mu}{(h^{\dag})^{\mu}}_{\nu}.
\end{equation} 
As such, we may also denote it as $-2\mathrm{div}_{g}^{\flat}$ whenever there is ambiguity about tensor type. We may sometimes suppress the factor of $-2$ for certain computations. The musical isomorphisms commute with the covariant derivative because $\nabla g = 0$. To describe the $L_{\infty}$ structure beyond $\ell_{1}$, we shift gears to the action functional perspective. We must introduce the following object, for which additional details can be found in \cite{sunryan}.

\begin{defn}
	Let $\mathscr{L}$ be a local $L_{\infty}$ algebra on $M$. $\mathscr{L}$ is \textbf{cyclic of degree $k$} if it is equipped with a degree $k$ non-degenerate symmetric pairing of vector bundles 
	$$
	\langle -,- \rangle_{L} : L \otimes L \to \mathrm{Dens}_{M}[k]
	$$
	such that the induced integration pairing on compactly supported sections $\langle -,- \rangle : \mathscr{L}_{c} \otimes \mathscr{L}_{c} \to \mathbf{R}$ is invariant; i.e., $\mathscr{L}_{c}^{\otimes(n+1)} \to \mathbf{R}$ defined as $\alpha_{1} \otimes \cdots \otimes \alpha_{n+1} \mapsto \langle \ell_{n}(\alpha_{1}, \ldots, \alpha_{n}), \alpha_{n+1} \rangle$ is graded antisymmetric in all the $\alpha_{i}$.
\end{defn}

The \textit{BV action functional} $S_{BV}$, which encodes the perturbative equations of motion, gauge invariance, and the action of vector fields on the dynamical fields and on the vector fields themselves via the Lie bracket, allows us to read off from it the $L_{\infty}$ structure, since the $L_{\infty}$ algebra encoding a BV theory is cyclic of degree $-3$.\footnote{This is due to \cite{cosgwill2}: see Sections 4.2 and 4.3.} The $S_{BV}$ of interest for us has the form
\begin{equation}\label{BV Lagrangian}
	S_{BV} = \langle h, D\mathbf{E}_{g}(h) \rangle + \sum_{k=2}^{\infty} \frac{1}{k!} \langle h, D^{k}\mathbf{E}_{g}(h, \ldots, h) \rangle  + \langle h^{\dag}, L_{X}g \rangle + \langle h^{\dag}, L_{X}h \rangle + \langle X^{\dag}, [X,X] \rangle,
\end{equation}
where each integration is with respect to the volume form $\mathrm{vol}_{g}$, as shown above. The infinite sum of the second summand contains all higher order variations of the nonlinear operator $\mathbf{E}$ defining the theory. In this set-up, $\mathbf{E}$ itself is derived from the variation of an action functional, so that $D\mathbf{E}_{g}$ results from the second variation, $D^{2}\mathbf{E}_{g}$ from the third, and so on: these define higher $L_{\infty}$ brackets $\ell_{n}$. In principle, we allow for infinitely many $L_{\infty}$ brackets because many gravitational theories of interest admit infinitely many variations, in contrast to gauge theories like Chern-Simons or Yang-Mills; however, in practice, we truncate at finitely many. More will be said about this in the remainder of this section and in Section \ref{factorizationalgebras}.

\begin{rmk}
The $\ell_{2}(X,X) = [X,X]$ term in $S_{BV}$ \textit{looks like} it should be $0$, since vector fields self-commute, but $S_{BV}$ is in fact a polynomial function of $\mathscr{E}_{g}$, given that it is in $\mathrm{Sym}(\mathscr{E}_{g}^{\vee}[-1])$.  The shift in degree implies that $[X,X] = XX + XX = 2X^{2}$, so that the term sticks around. It encodes the action of vector fields on themselves and on anti-vector fields in dual degree. More detail on this can be found in Section 3.1 of \cite{bsw}. 

Moreover, that $\mathscr{E}_{g}$ defines a DGLA comes from the fact that $\mathfrak{K}_{g}$ was a DGLA--which we showed above--and we are extending it further by modules of $\mathrm{Vect}_{M}$ in higher degrees. $\mathscr{E}_{g}$ is an $L_{\infty}$ algebra for similar reasons: the only nontrivial $n$-ary brackets for $n \geq 3$ are those derived from the higher order variations of the nonlinear action functional--those in the second summand of Equation (\ref{BV Lagrangian}).
\end{rmk}

\end{const}

\begin{rmk}
Any generally covariant vacuum gravity theory should define a perturbative BV theory described by the above data: a differential graded space of fields, a shifted symplectic pairing, and a BV action functional prescribing the cyclic $L_{\infty}$ structure. The diffeomorphism equivariance of the equations of motion guarantees that the DG fields underlying $\mathscr{E}_{g}$ define a cochain complex. But we need to look at specific models to make meaningful statements. We start with general relativity before looking at conformal gravity to consider some interesting differences.
\end{rmk}

\noindent \textbf{Note:} We assume from now on that the dimension of our spacetime $M$ is 4.

\subsection{Perturbative GR}\label{einstein} For this model, we use the space of fields $\mathscr{E}_{g}$ from Equation (\ref{BVfields}), and choose $D\mathbf{E}_{g}(h) = 0$ to be the linearized vacuum Einstein equation $D\mathrm{Ein}_{g}(h) = 0$ at a fixed metric $g$. This is equivalent to $D\mathrm{Ric}_{g}(h) = 0$, where $D\mathrm{Ric}_{g}$ is the linearized Ricci tensor at $g$, for the same reason that the Einstein equations reduce to $\mathrm{Ric}(g) = 0$ in vacuum. To be more specific, we choose our operator to be $D\mathrm{Ric}_{g}^{\sharp} := \sharp \circ D\mathrm{Ric}_{g}$ so that its codomain is of the correct tensor type: composing with tensor index raising or lowering does not effect the gauge invariance.

We may choose to impose ``transverse traceless" gauge coming from the gauge freedom $h \mapsto h + L_{X}g$: namely, we demand that $\mathrm{div}_{g}(Ih) = 0$, where $I$ denotes the trace-reversal operation $I : h_{\mu\nu} \mapsto h_{\mu\nu} - \frac{1}{2}hg_{\mu\nu}$, where $h = \mathrm{tr}_{g}(h)$. It is shown in Equation 7.5.23 in \cite{wald} that $D\mathrm{Ric}_{g}(h) = 0$ in this gauge may be written 
\begin{equation}\label{linearizedricci}
	(-\square^{\nabla} + 2\mathrm{Riem})(h) := -\nabla^{\alpha}\nabla_{\alpha}h_{\mu\nu} + 2 {{R^{\alpha}}_{\mu\nu}}^{\beta}h_{\alpha\beta} = 0,\footnote{We suppress the raising of indices between $\mathscr{E}_{g}$ in degrees 1 and 2 unless necessary to deal with directly.}
\end{equation}
where the covariant derivative $\nabla$ and the Riemann tensor are associated to the background metric $g$, and the symbol $\mathrm{Riem}$ denotes the Riemann tensor as an endomorphism of symmetric $(0,2)$ tensors. If the background metric is flat, then the linearized equation reduces to the wave equation $\square h_{\mu\nu} = 0$ for the components of the metric perturbation. 

\begin{rmk}
Even if the background metric is not flat, it must still be Ricci flat, thanks to the vacuum Einstein equations. In other words, the ``source" of the propagation in Equation (\ref{linearizedricci}) becomes the \textit{traceless} part of the Riemann tensor, known as the \textit{Weyl tensor}.
\end{rmk}

For our purposes, we must consider the linearized equation of motion  without imposing gauge conditions. According to Equation 2.9 in \cite{musante}, it is
\begin{equation}\label{fulllinearizedequation}
	D\mathrm{Ric}_{g}(h) = (-\square^{\nabla} + 2\mathrm{Riem} + 2I(L_{\bullet^{\sharp}}g)\mathrm{div}_{g})(Ih) = 0,
\end{equation}
where $L_{\bullet^{\sharp}}g$ sends $\alpha \in \Omega^{1}_{M}$ to $L_{\alpha^{\sharp}}g \in \Gamma(\mathrm{Sym}^{2}T^{*}_{M})$. Importantly, $I$ is an involution on $\Gamma(\mathrm{Sym}^{2}T^{*}_{M})$ in dimension 4, and is thus an isomorphism, and it commutes with $-\square^{\nabla} + 2\mathrm{Riem}$, as shown in \cite{musante}. The final summand on the left side obstructs the linear operator from being Green hyperbolic, which is equivalent to it admitting retarded and advanced Green's operators. However, we are in luck, since the final term is simply a result of the aforementioned gauge freedom: more will be said on this in Section \ref{factorizationalgebras}. For the time being, we use the operator in Equation (\ref{fulllinearizedequation}).

\begin{prop}\label{einsteinisBV}
	The data $(\mathscr{E}_{g}, S_{BV}, \omega)$ defined in Construction \ref{BVtheory} in the case of perturbative general relativity around a fixed background Einstein metric $g \in [\mathrm{Lor}(M)/\mathrm{Diff}^{+}(M)]$, where $D\mathbf{E}_{g}(h) := D\mathrm{Ric}_{g}^{\sharp}(h)$, defines a BV classical field theory.
\end{prop}
\begin{proof}
That the fields $\mathscr{E}_{g}$ define a cochain complex may be manually checked by substituting $h = L_{X}g$ into Equation (\ref{fulllinearizedequation}): this is a quick exercise in unpacking the definition of ${{R^{\alpha}}_{\mu\nu}}^{\beta}$. One may also invoke that $L_{X}\mathrm{Ric}(g) = D\mathrm{Ric}_{g}(L_{X}g)$, so that if $\mathrm{Ric}(g) = 0$, we have the result.\footnote{Our operator is modified to land in contravariant $2$-tensors, but that this does not effect gauge invariance.} This and related formulas regarding the gauge invariance of general relativity may be found in Section 1 of \cite{fmm}. Finally, $\mathrm{div}_{g}(D\mathrm{Ric}_{g}^{\sharp}(h)) = 0$ is simply the contracted linearized Bianchi identity, as shown in Lemma 1.1 of \cite{fmm}. Letting $S_{BV}$ be as in (\ref{BV Lagrangian}) with the preceding choice for $D\mathbf{E}_{g}$ and the shifted symplectic form $\omega$ be $\langle-,-\rangle$ defined in (\ref{shiftedsymplectic}), a routine computation coming from the gauge invariance of the theory under the action of vector fields shows that $\{S_{BV},S_{BV}\} = 0$.
\end{proof}

\subsection{Perturbative Conformal Gravity}\label{weyl} A leading contender among alternative formulations of gravity is conformal gravity, sometimes called Weyl gravity due to the form of its action functional. An introduction to the theory with a thorough description of its relevance to the dark matter problem is provided in \cite{mannheim}. However, it is also of mathematical interest as a conformal field theory. Its actional functional is 
\begin{equation}
	S_{W} = \int_{M} |W|^{2} \mathrm{vol}_{g} := \int_{M} W_{\mu\nu\alpha\beta}W^{\mu\nu\alpha\beta} \mathrm{vol}_{g},
\end{equation}
where the \textit{Weyl tensor} is defined as the traceless part of the Riemann tensor: $W_{\mu\nu\alpha\beta} := R_{\mu\nu\alpha\beta} + (n-2)^{-1}(R_{\mu\beta}g_{\nu\alpha} - R_{\mu\alpha}g_{\nu\beta} + R_{\nu\alpha}g_{\mu\beta} - R_{\nu\beta}g_{\mu\alpha}) + (n-1)^{-1}(n-2)^{-1}R(g_{\mu\alpha}g_{\nu\beta} - g_{\mu\beta}g_{\nu\alpha})$. 
The Euler-Lagrange equation for this action functional is
\begin{equation}\label{bach}
	B_{\mu\nu} := \nabla^{\alpha}\nabla^{\beta}W_{\alpha\mu\nu\beta} + \frac{1}{2}R^{\alpha\beta}W_{\alpha\mu\nu\beta} = 0,
\end{equation}
where $B_{\mu\nu}$ is the \textit{Bach tensor} and the equation is thus called the Bach equation. In dimension 4, a metric for which $B_{\mu\nu} = 0$ is conformally equivalent to an Einstein metric. A coordinate-free expression for this equation is Equation 4.77 in the wonderful book \cite{besse}.

\begin{rmk}[Conformal Invariance]\label{conformalinvariance}
The Weyl tensor in its $(1,3)$ form ${W^{\alpha}}_{\mu\nu\beta}$ is a conformal invariant of \textit{weight} $0$ in every dimension, in the sense that under a conformal change of metric $g \mapsto e^{2f}g$ for $f \in C^{\infty}(M)$, ${W^{\alpha}}_{\mu\nu\beta}$ remains unchanged. The action functional above is similarly conformally invariant. For $\dim M = 4$, the $(0,2)$ Bach tensor is a conformal invariant of weight $-2$, in the sense that $B(e^{2f}g) = e^{-2f}B(g)$. However, $B$ is not a conformal invariant of any weight in any other dimension. Moreover,  $B$ is divergenceless only in dimension 4. As such, when studying conformal gravity, we always fix $\dim M = 4$. 
\end{rmk}

\begin{rmk}
	It is a short but worthwhile exercise to show that any solution to the nonlinear vacuum Einstein equation $R_{\mu\nu} = 0$ is also a solution to the vacuum Bach equation. In other words, $\mathrm{Crit}(S_{EH}) \hookrightarrow  \mathrm{Crit}(S_{W})$ as loci in the moduli stack  $[\mathrm{Lor}(M)/\mathrm{Diff}^{+}(M)]$, where $S_{EH}$ is the Einstein-Hilbert functional of general relativity. However, the converse is not true: in Section 9 of \cite{mannheim}, the author describes a Schwarzschild-like solution to Equation (\ref{bach}) where the potential has the form $V(r) = 1 - 2br^{-1} + cr$. This spherically symmetric and static solution is \textit{not} a solution to the vacuum Einstein equations. 
\end{rmk}

Due to the conformal symmetries described above, it is necessary for us to include their infinitesimal manifestations in our description of the associated formal moduli problem. Following 12.6.2 in \cite{cosgwill2}, we modify the DGLA $\mathfrak{K}_{g}$ from earlier by adding infinitesimal conformal transformations (also known as infinitesimal \textit{Weyl rescalings}) into degree $0$: 
\begin{equation}
	\mathfrak{L}_{g} := 0 \to \mathrm{Vect}_{M} \oplus C^{\infty}_{M} \xrightarrow{\rho_{g}} \Gamma (\mathrm{Sym}^{2}T_{M}^{*})[-1] \to 0.
\end{equation}
The differential is $\rho_{g}(X,f) := L_{X}g + 2fg$. This inherits the DGLA structure from $\mathfrak{K}_{g}$ as long as $[f,h] := 2 fh$ for $h \in \Gamma (\mathrm{Sym}^{2}T_{M}^{*})$ and that $C^{\infty}_{M}$ does not act on $\mathrm{Vect}_{M}$. The constant $2$ is chosen so that the exponentation of the infinitesimal transformation gives back the conformal factor $e^{2f}$.

\begin{lem}
There is a natural inclusion	 $\mathfrak{K}_{g} \hookrightarrow \mathfrak{L}_{g}$ of DGLAs.
\end{lem}
\begin{proof}
That $\mathfrak{L}_{g}$ is in fact a DGLA is a somewhat tedious computation, most of it following completely analogously to what is shown in Equation (\ref{dglacomputation})--noting that the Lie bracket is $[(X,f),(X',f')] = ([X,X'], 0)$, because the infinitesimal conformal transformations are an abelian Lie subalgebra. 

Next, the above inclusion is $\iota(X) = (X,0)$ in degree 0 and the identity in degree 1. Thus, $\rho_{g}(\iota(X)) = L_{X}g$, so that we have a cochain inclusion. This same observation implies that Equation (\ref{dglacomputation}) is sufficient to conclude that the DGLA structure of $\mathfrak{K}_{g}$ maps injectively into that of $\mathfrak{L}_{g}$.
\end{proof}

\begin{rmk}
	Note that if $\rho_{g}(X,f) = 0$, then $L_{X}g = -2fg$. In this case, it can be shown that $-2f = \frac{1}{2} \mathrm{div}_{g}(X)$,\footnote{One must be mindful that the constants are a bit different in dimension $n \neq 4$.} and $X$ is called a \textit{conformal Killing field}. Ergo, Killing fields are special cases of conformal Killing fields where $\mathrm{div}_{g}(X) = 0$, which is how we get $H^{0}(\mathfrak{K}_{g}) \subseteq H^{0}(\mathfrak{L}_{g})$. 
\end{rmk}

The underlying cochain complex of the formal moduli problem of perturbative conformal gravity around a fixed metric $g$ in the moduli stack $[\mathrm{Lor}(M)/\mathrm{Diff}^{+}(M)]$ such that $B(g) = 0$ is thus $T^{\vee}[-1]\mathfrak{L}_{g}$ twisted by the linearized Bach tensor $DB_{g}$ between the degree 1 and 2 parts. We denote it as
\begin{equation}
	\mathscr{W}_{g} := \mathrm{Vect}_{M} \oplus C^{\infty}_{M} \xrightarrow{\rho_{g}} \Gamma (\mathrm{Sym}^{2}T_{M}^{*})[-1] \xrightarrow{DB_{g}} \Gamma(\mathrm{Sym}^{2}T_{M})[-2] \xrightarrow{\rho_{g}^{*}} {(\Omega^{1}_{M} \oplus C^{\infty}_{M})}[-3].
\end{equation}
Note that the $L^{2}$-dual of $f \mapsto 2fg$ is $h^{\dag} \mapsto 2\mathrm{tr}_{g}(h^{\dag})$, and this gets incorporated into the definition of $\rho_{g}^{*}$. That $\mathscr{W}_{g}$ is a DGLA follows from similar reasons to the case of $\mathscr{E}_{g}$: namely, that we extended $\mathfrak{L}_{g}$ by modules of the action $\rho_{g}$. Most of the proof that the above defines a classical BV theory follows identically to the case of perturbative GR; however, it remains to prove one crucial fact.

\begin{lem}
$\mathrm{im}(\rho_{g}) \subset \ker(DB_{g})$.	
\end{lem}
\begin{proof}
	This follows from the conformal invariance property of $B$ described in Remark \ref{conformalinvariance}, but we should be careful to show it in some detail. We must check that for any $(X,f) \in \mathrm{Vect}_{M} \oplus C^{\infty}_{M}$, it must be that $DB_{g}(\rho_{g}(X,f)) = DB_{g}(L_{X}g + 2fg) = DB_{g}(L_{X}g) + DB_{g}(2fg) = 0$. 
	
	It follows from the usual diffeomorphism covariance that $DB_{g}(L_{X}g) = 0$, so we must show that $DB_{g}(2fg) = 0$. To do this, we let $\{ e^{2tf} \}_{t \geq 0}$ be a one-parameter group of conformal transformations. Using $B(e^{2f}g) = e^{-2f}B(g)$ from earlier, we see that $B(e^{2tf}g) = e^{-2tf}B(g)$. Computing the derivative with respect to $t$ of both sides and setting $t=0$, we arrive at $DB_{g}(2fg) = -2fB(g)$. However, we assumed that $B(g) = 0$, so that $DB_{g}(2fg) = 0$.
\end{proof}

This proof makes it clear just how important it is to assume Bach-flatness of the fixed background metric. To bring the main proposition of this section home, we must slightly generalize the form of the earlier generic BV action functional. To include the symmetry from the infinitesimal conformal transformations, $S_{BV}$ is adapted from Equation (\ref{BV Lagrangian}) to be
\begin{equation}\label{WeylBVfunctional}
		\langle h, DB_{g}(h) \rangle + \sum_{k=2}^{\infty} \frac{1}{k!} \langle h, D^{k}B_{g}(h, \ldots, h) \rangle  + \langle h^{\dag}, L_{X}g \rangle + \langle h^{\dag}, L_{X}h \rangle + \langle X^{\dag}, [X,X] \rangle + \langle f^{\dag}, [X,f] \rangle.
\end{equation}
More about encoding the action of infinitesimal conformal transformations with this flavor of action functional is in Section 2.3 of \cite{gwilliams}. Thanks to the preceding lemmas and arguments mirrored from the proof for Proposition \ref{einsteinisBV}, we may conclude: 

\begin{prop}\label{weylisBV}
	The data $(\mathscr{W}_{g}, S_{BV}, \omega)$ defined above in the case of perturbative conformal gravity around a fixed background Bach-flat metric $g \in [\mathrm{Lor}(M)/\mathrm{Diff}^{+}(M)]$, where $D\mathbf{E}_{g}(h) := DB_{g}(h)$ and the higher $L_{\infty}$ brackets are as provided in (\ref{WeylBVfunctional}), defines a BV classical field theory.
\end{prop}

\begin{rmk}
	The exact and somewhat unwieldy form of the  linearized Bach tensor around an arbitrary metric is shown in Equation (43) of \cite{apm}, and simplifications around conformally Einstein and conformally flat metrics are also provided there. The authors show that in transverse traceless gauge and in a conformally flat Einstein background, the linearized Bach equation reduces to $\square^{2} h_{\mu\nu} = 0$, a fourth order equation in the perturbation.\footnote{This is expected, as the Bach equation itself is fourth order in the metric.} Due to this and the natural inclusion $\mathfrak{K}_{g} \hookrightarrow \mathfrak{L}_{g}$ of $L_{\infty}$ algebras, we get that:
\end{rmk}

\begin{prop}\label{BVinclusion}
	When the background metric $g \in [\mathrm{Lor}(M)/\mathrm{Diff}^{+}(M)]$ is assumed to be conformally flat and Einstein, then there is a natural inclusion $\mathscr{E}_{g} \hookrightarrow \mathscr{W}_{g}$ of $L_{\infty}$ algebras representing BV perturbative general relativity and BV perturbative conformal gravity.
\end{prop}

This is not surprising in light of the fact that $\mathrm{Crit}(S_{EH}) \hookrightarrow  \mathrm{Crit}(S_{W})$ as loci in the moduli stack  $[\mathrm{Lor}(M)/\mathrm{Diff}^{+}(M)]$. Indeed, it follows immediately from the assumption that $g$ is both conformally flat and Einstein by the work of \cite{apm}. But we expect it to be true for a background metric $g$ which is just Einstein: we shall leave it to future work to show that Proposition \ref{BVinclusion} holds in that case. For now, we set that aside and move on to the final section.

\section{The Factorization Algebras of Observables}\label{factorizationalgebras}

\begin{rmk}
	If the background metric $g$ were Riemannian and Ricci flat, then the observables of the BV theories of the preceding section would already define factorization algebras on the subsite $\mathbf{RicRiem}$ of Riemannian manifolds which are Ricci flat, by Proposition 2.43 of \cite{dul}. However, the Lorentzian structure is essential for our purposes. As such, we now describe which sacrifices must be made (for now) to make sense of the Lorentzian version. 
	
	Moreover, we must note here that in contrast to the results of \cite{bms}, the category $\mathbf{Loc}$ must be restricted to the subcategory $\mathbf{RicLoc}$ of Ricci flat globally hyperbolic spacetimes. For linearized gravity to define a BV theory in the sense defined in the preceding article, the background metric must satisfy the vacuum Einstein equations: otherwise, there would be an obstruction to defining the necessary complex of linear differential operators, since the linearized Bianchi identity would not hold. This is a significant difference as compared to other theories defined on $\mathbf{Loc}$.  
\end{rmk}

Now that we know how BV perturbative gravity works out in the case of traditional general relativity and the more iconoclastic conformal gravity, we will take a look at the observables of the former, which define a factorization algebra on the site $\mathbf{RicLoc}$ as long as we impose a few constraints: these allow us to invoke the results in \cite{bms}. We must truncate the DGLAs representing the BV theories to their underlying cochain complexes: this means eliminating the terms with $[X,X]$ from the BV action functionals. Moreover, we must shift the field spaces down by one, so that the fields are concentrated between degrees $-1$ and 2; however, we will continue to refer to the now \textit{free} BV theories as $\mathscr{E}_{g}$ and $\mathscr{W}_{g}$.

 Our goal is to see how starting with the formal moduli problem definition of a classical BV theory we used in the preceding section and then projecting it to its underlying cochain complex (while retaining the shifted symplectic structure) gives us a free BV theory in the sense of Definition 3.5 of \cite{bms}. This amounts to checking a few properties, foremost among them whether our resulting cochain complex admits a \textit{Green's witness}: a sequence of maps going down in degree which detects whether or not the cochain complex in question is \textit{Green hyperbolic}. From here, we consider two distinct Poisson structures on the linear observables which extend to the symmetric algebra of the linear observables to define a factorization algebra on $\mathbf{RicLoc}$. 
 
 \subsection{Green's Witnesses and Poisson Structures} 
 
 The cochain complexes $\mathscr{E}_{g}$ and $\mathscr{W}_{g}$ from earlier define complexes of linear differential operators in the sense of Section 3.1 in \cite{bms}:
 
 \begin{defn}
 	A \textbf{complex of linear differential operators} $(F,Q)$ over $M \in \mathbf{Loc}$ consists of a graded vector bundle $F \to M$ and a collection $Q = (Q^{n} : \mathscr{F}(M)^{n} \to \mathscr{F}(M)^{n+1})$ of degree increasing linear differential operators such that $Q^{n+1}Q^{n} = 0$ for all $n \in \mathbf{Z}$, where $\mathscr{F}(M) \in \mathbf{Ch}$ denotes the cochain complex of sections associated with $(F,Q)$. 
 \end{defn}
 
 \begin{rmk}
 Moreover, the shifted symplectic pairing of Equation (\ref{shiftedsymplectic}) is related to what in \cite{bms} is called a \textbf{compatible $(-1)$-shifted fiber metric}: a fiber-wise non-degenerate, graded anti-symmetric, graded vector bundle pairing $(-,-) : F \otimes F \to M \times \mathbf{R}[-1]$ such that 
 \begin{equation}
 	\int_{M} (Q\varphi_{1}, \varphi_{2}) \mathrm{vol}_{g} + (-1)^{|\varphi_{1}|} \int_{M} (\varphi_{1}, Q\varphi_{2}) \mathrm{vol}_{g} = 0
 \end{equation}
 for all homogeneous sections $\varphi_{1}, \varphi_{2} \in \mathscr{F}(M)$ with compact overlapping support. In the guise of (\ref{shiftedsymplectic}), this equation is $\langle Q\varphi_{1}, \varphi_{2} \rangle + (-1)^{|\varphi_{1}|}\langle \varphi_{1}, Q \varphi_{2} \rangle = 0$. This \textit{compatibility condition} implies that $\langle -,- \rangle : \mathscr{F}_{c}(M) \otimes \mathscr{F}(M) \to \mathbf{R}[-1]$ is a cochain map.\footnote{In \cite{bms}, this integration pairing is denoted $((-,-))$.}
 \end{rmk}
 
 That both $\mathscr{E}_{g}$ and $\mathscr{W}_{g}$ satisfy $\langle Q\varphi_{1}, \varphi_{2} \rangle + (-1)^{|\varphi_{1}|}\langle \varphi_{1}, Q \varphi_{2} \rangle = 0$ for their analogously defined shifted symplectic pairings is evident from the definition of the underlying cochain complexes. For example, in the simpler case of $\mathscr{E}_{g}$, we see that 
 \begin{equation}\label{Lieadjoint}
 \langle L_{X}g, h^{\dag} \rangle + (-1)^{|X|}\langle X, -2\mathrm{div}_{g}^{\flat}(h^{\dag}) \rangle  = 0,
 \end{equation}
 from the definition of the dual of $L_{\bullet}g$, and that for $h$ and $h'$ in degree $0$, 
 \begin{equation}\label{DRicadjoint}
 \langle D\mathrm{Ric}_{g}(h), h' \rangle + (-1)^{|h|} \langle h, D\mathrm{Ric}_{g}(h') \rangle  = 0
\end{equation}
 by the formal self-adjointness of the linearized Ricci tensor. The case of linearized conformal gravity follows suit, except that $L_{X}g$ is replaced with $L_{X}g + 2fg$, where the dual of the second term satisfies
 \begin{equation}\label{conformaladjoint}
 	\langle 2fg, h^{\dag} \rangle + (-1)^{|f|}\langle f, 2\mathrm{tr}_{g}(h^{\dag}) \rangle = 0.
 \end{equation}
To introduce the \cite{bms} notion of a free BV theory, we introduce one more ingredient. 

\begin{defn}\label{greenswitness}
	A \textbf{(formally self-adjoint) Green's witness} $W = (W^{n})_{n \in \mathbf{Z}}$ for a complex of linear differential operators $(F,Q)$ endowed with a compatible $(-1)$-shifted fiber metric $(-,-)$ consists of a collection of degree decreasing linear differential operators $W^{n} : \mathscr{F}(M)^{n} \to \mathscr{F}(M)^{n-1}$ such that the following conditions hold: 
	(i) for all $n \in \mathbf{Z}$, $P^{n} := Q^{n-1}W^{n} + W^{n+1}Q^{n} : \mathscr{F}(M)^{n} \to \mathscr{F}(M)^{n}$ are Green hyperbolic operators; (ii) $QWW = WWQ$; (iii) for all homogeneous sections $\varphi_{1}, \varphi_{2} \in \mathscr{F}(M)$ with compact overlapping support, $\langle W\varphi_{1}, \varphi_{2}\rangle = (-1)^{|\varphi_{1}|}\langle \varphi_{1}, W\varphi_{2} \rangle$. 
\end{defn}

\begin{rmk}\label{consequences}
Some important consequences of the existence of such a Green's witness are as follows.
(1) For all integers $n$, there exist unique retarded and advanced Green's operators $G^{n}_{\pm} : \mathscr{F}_{c}(M)^{n} \to \mathscr{F}(M)^{n}$ associated with the Green hyperbolic operators $P^{n}$. (2) We have that $PW = WP$ and $PQ = QP$, and so that $G_{\pm}W = WG_{\pm}$. (3) $P$ is formally self adjoint: $\langle P\varphi_{1}, \varphi_{2} \rangle = \langle \varphi_{1}, P\varphi_{2} \rangle$ for all homogeneous sections $\varphi_{1}, \varphi_{2} \in \mathscr{F}(M)$ with compact overlapping support. Moreover, $\langle G_{\pm}\psi_{1}, \psi_{2} \rangle = \langle \psi_{1}, G_{\mp}\psi_{2} \rangle$ for all $\psi_{1}, \psi_{2} \in \mathscr{F}_{c}(M)$, which implies that the retarded-minus-advanced propagator $G := G_{+} - G_{-}$ and Dirac propagator $G_{D} := \frac{1}{2}(G_{+} + G_{-})$ satisfy $\langle G\psi_{1}, \psi_{2} \rangle = -\langle \psi_{1}, G\psi_{2} \rangle$ and $\langle G_{D}\psi_{1}, \psi_{2} \rangle = \langle \psi_{1}, G_{D}\psi_{2} \rangle$ for all $\psi_{1}, \psi_{2} \in \mathscr{F}_{c}(M)$.
\end{rmk}

With the above in hand, we may introduce the central definition of this section: 

\begin{defn}[Adapted from 3.5 in \cite{bms}]\label{bmsBVtheory} (i) A \textbf{free BV theory} $(F,Q, (-,-),W)$ on $M \in \mathbf{Loc}$ consists of a complex of linear differential operators $(F,Q)$ with a compatible $(-1)$-shifted fiber metric $(-,-)$ and a Green's witness $W$. (ii) A \textbf{free BV theory} $(F,Q, (-,-),W)$ on $M \in \mathbf{RicLoc}$ is a free BV theory for $M \in \mathbf{Loc}$ which are Ricci flat.
\end{defn}

To continue, we reintroduce the differential graded structure of $\mathscr{E}_{g}$ in addition to being clear about the new grading. We make the tensor index lowering and raising explicit by using $\flat$ for index lowering and $\sharp$ for index raising, with both the fixed metric $g$ and the tensor type implicit: for example, if $A$ is a $(k,0)$ tensor, then $A^{\flat}$ is a $(0,k)$ tensor. This will be important when introducing our Green's witness. With the above in hand, we reiterate that from now on,
\begin{equation}\label{BVGR}
	\mathscr{E}_{g}:= 0 \to  \mathrm{Vect}_{M}[1] \xrightarrow{L_{\bullet}g} \Gamma(\mathrm{Sym}^{2}T_{M}^{*}) \xrightarrow{D\mathrm{Ric}_{g}^{\sharp}} \Gamma(\mathrm{Sym}^{2}T_{M})[-1] \xrightarrow{-2\mathrm{div}_{g}^{\flat}} {\Omega^{1}_{M}}[-2] \to 0,
\end{equation}
where $\mathrm{div}_{g}^{\flat}:= \flat \circ \mathrm{div}_{g}$ and $D\mathrm{Ric}_{g}^{\sharp} := \sharp \circ D\mathrm{Ric}_{g}$, where we recall that $D\mathrm{Ric}_{g}$ is $(-\square^{\nabla} + 2\mathrm{Riem} + 2I(L_{\bullet^{\sharp}}g)\mathrm{div}_{g}) \circ I$, harkening back to Equation (\ref{fulllinearizedequation}). 

\begin{rmk}
Using the complete linearized Ricci tensor of Equation (\ref{fulllinearizedequation}) and choosing a Green's witness $W$ susses out the underlying Green hyperbolic operator of linearized gravity, as seen in Equation (\ref{linearizedricci}), as $P = QW + WQ$. Thus, it is an advantage to leave $Q$ in terms which at first seem to compromise Green hyperbolicity, since the cohomological setup and the choice of a Green's witness $W$ identifies the underlying Green hyperbolic operator $P$ of the theory in a natural  way.
\end{rmk}

\begin{thm}\label{greenswitnessGR}
	A formally self-adjoint Green's witness for $\mathscr{E}_{g}$ as defined in (\ref{BVGR}) is 
	\begin{equation}
		0 \xleftarrow{} \mathrm{Vect}_{M}[1] \xleftarrow{-2\mathrm{div}_{g}^{\sharp} \circ I} \Gamma(\mathrm{Sym}^{2}T_{M}^{*}) \xleftarrow{\flat \circ I} \Gamma(\mathrm{Sym}^{2}T_{M})[-1] \xleftarrow{I \circ L^{\sharp}_{\bullet^{\sharp}}g} {\Omega^{1}_{M}}[-2] \xleftarrow{} 0,
	\end{equation} 
	where the part from degree 2 to degree 1 sends $\alpha \in \Omega^{1}_{M}$ to $I \circ (L_{\alpha^{\sharp}}g)^{\sharp}$. 
\end{thm}
\begin{proof}
	Let us start with point (i) of Definition \ref{greenswitness}. Computing, we get $P^{-1} = Q^{-2}W^{-1} + W^{0}Q^{-1} = (-2\mathrm{div}_{g}^{\sharp} \circ I)(L_{\bullet}g)$ and $P^{2} = Q^{1}W^{2} + W^{3}Q^{2} = (-2\mathrm{div}_{g}^{\flat})(I \circ L^{\sharp}_{\bullet^{\sharp}}g)$, since $Q^{-2}W^{-1}$ and $W^{3}Q^{2}$ are both zero. $P^{-1}$ and $P^{2}$ are the d'Alembertians $\square$ on $\mathrm{Vect}_{M}$ and $\Omega^{1}_{M}$, respectively, as shown in Lemma 2.1.13 of \cite{musante}. Thus, they are both Green hyperbolic. 
	
	Next, we have $P^{0} = (L_{\bullet}g)(-2\mathrm{div}_{g}^{\sharp} \circ I) + (\flat \circ I)(D\mathrm{Ric}_{g}^{\sharp})$. By expanding $\flat \circ I \circ D\mathrm{Ric}_{g}^{\sharp} = \flat \circ I \circ (-\square^{\nabla} + 2\mathrm{Riem} + 2I(L_{\bullet^{\sharp}}g)\mathrm{div}_{g})^{\sharp} \circ I$, we can see that 
	$$
	P^{0} = (L_{\bullet}g)(-2\mathrm{div}_{g}^{\sharp} \circ I) + I \circ (-\square^{\nabla} + 2\mathrm{Riem}) \circ I + 2I^{2}(L_{\bullet^{\sharp}}g)\mathrm{div}_{g} \circ I.
	$$
This reduces to
$$
P^{0} = -2(L_{\bullet}g)\mathrm{div}_{g}^{\sharp} \circ I + (-\square^{\nabla} + 2\mathrm{Riem}) + 2(L_{\bullet^{\sharp}}g)\mathrm{div}_{g} \circ I = -\square^{\nabla} + 2\mathrm{Riem},
$$
because $\square^{\nabla} + 2\mathrm{Riem}$ commutes with $I$ (also by Lemma 2.1.13 in \cite{musante}), $I$ is an involution, so that $I^{2} = 1$, and keeping track of tensor types shows that the first and third summands cancel. Since $P^{0} = -\square^{\nabla} + 2\mathrm{Riem}$ acting on $\Gamma(\mathrm{Sym}^{2}T^{*}_{M})$ is a normally hyperbolic operator, it is therefore Green hyperbolic. Finally, $P^{1} = (D\mathrm{Ric}_{g}^{\sharp})(\flat \circ I) + (I \circ L^{\sharp}_{\bullet^{\sharp}}g)(-2\mathrm{div}_{g}^{\flat})$. By the same expansion with a mirrored computation regarding tensor types, $P^{1}$ also reduces to $-\square^{\nabla} + 2\mathrm{Riem}$, but this time acting on $\Gamma(\mathrm{Sym}^{2}T_{M})$, so it is Green hyperbolic in a similar way.
	
	To show (ii), we must check if the degree $-1$ maps $QWW$ and $WWQ$ agree in each degree. That the degree $0$ to degree $-1$ parts agree follows from the fact that $WWQ = (-2\mathrm{div}_{g}^{\sharp})(\flat \circ I)(D\mathrm{Ric}_{g}^{\sharp}) = 0$, as it is equivalent to the contracted linearized Bianchi identity, and $QWW = 0$ identically. From degree 1 to 0, $WWQ = (\flat \circ I)(L_{\bullet^{\sharp}}^{\sharp}g)(-2\mathrm{div}_{g}^{\flat})$ and $QWW = (L_{\bullet}g)(-2\mathrm{div}_{g}^{\sharp})(\flat \circ I)$ are equivalent after some tensor type bookkeeping. The degree 2 to degree 0 mapping is mirrored from the degree 0 to degree $-1$ case, so we omit it, and the rest of the instances are identically 0. 
	
	Finally, we remark that (iii) follows by similar reasoning to what was shown in Equations (\ref{Lieadjoint}) and (\ref{DRicadjoint}), so we omit that repetition also. Since our choice of $W$ satisfies properties (i) through (iii) outlined in Definition \ref{greenswitness}, it defines a Green's witness, as desired.
\end{proof}

\begin{cor}
	The choice of $\mathscr{E}_{g}$ from (\ref{BVGR}) in addition to the choice of Green's witness from Proposition \ref{greenswitnessGR} defines a free BV theory on $\mathbf{RicLoc}$ in the sense of Definition \ref{bmsBVtheory}.
\end{cor}

Central results of \cite{bms} now allow us to define a few Poisson structures on the observables, which then define analogous factorization algebras: for the sake of brevity, we refer to Sections 3.1 and 3.2 of that paper, but give a recap of the main ideas. We begin by noting that the preceding results imply, thanks to \cite{bmsgreen}, the existence of \textit{retarded/advanced Green's homotopies} $\Lambda_{\pm} := WG_{\pm} = G_{\pm}W \in [\mathscr{E}_{g,c}, \mathscr{E}_{g}]^{-1}$, where $[V,W]^{n}$ denotes the internal hom cochain complex $\prod_{p\in \mathbf{Z}} \mathrm{Hom}(V^{p}, W^{n+p})$. These trivialize the map $j : \mathscr{E}_{g,c} \to \mathscr{E}_{g}$ forgetting compact supports in the sense that $\partial \Lambda_{\pm} := Q\Lambda_{\pm} + \Lambda_{\pm}Q = QWG_{\pm} + WG_{\pm}Q = PG_{\pm} = j$. 

We further define the \textit{retarded-minus-advanced cochain map} $\Lambda := \Lambda_{+} - \Lambda_{-} = WG : \mathscr{E}_{g,c}[1] \to \mathscr{E}_{g}$ and the \textit{Dirac homotopy} $\Lambda_{D} := WG_{D} : \mathscr{E}_{g,c}[1] \to \mathscr{E}_{g}$. This also satisfies $\partial \Lambda_{D} = j$. The \textit{linear observables} in this case are identified with $\mathscr{E}_{g,c}[1]$, the compactly supported DG fields shifted down by 1. On the linear observables we define a $(-1)$-shifted Poisson structure $\tau_{(-1)} : \mathscr{E}_{g,c}[1] \to \mathbf{R}[1]$ by stating that for all homogeneous $\psi_{1},\psi_{2} \in \mathscr{E}_{g,c}[1]$,
\begin{equation}
	\tau_{(-1)}(\psi_{1} \otimes \psi_{2}) := \int_{M} (\psi_{1},\psi_{2})\mathrm{vol}_{g},
\end{equation}
where $(-,-)$ denotes the natural pairing between the two sections. The above suppresses a few details from the full definition given in Equation 3.13 of \cite{bms}, but it is sufficient for us. Note that $\tau_{(-1)}$ is graded symmetric. Similarly, the \textit{unshifted Poisson structure} $\tau_{(0)} : \mathscr{E}_{g,c}[1] \to \mathbf{R}$ is
\begin{equation}
	\tau_{(0)}(\psi_{1} \otimes \psi_{2}) := \int_{M}(\psi_{1}, WG \psi_{2})\mathrm{vol}_{g}
\end{equation}
and the \textit{Dirac pairing} $\tau_{D} : \mathscr{E}_{g,c}[1] \to \mathbf{R}$ (which is also unshifted) is 
\begin{equation}
	\tau_{D}(\psi_{1} \otimes \psi_{2}) := \int_{M}(\psi_{1}, WG_{D}\psi_{2})\mathrm{vol}_{g}.
\end{equation}
$\tau_{(0)}$ is graded anti-symmetric and $\tau_{D}$ is graded symmetric. It is also worth noting that $\tau_{D}$ trivializes $\tau_{(-1)}$, in the sense that $\partial \tau_{D} = \tau_{(-1)}$. 

\begin{rmk}\label{timeevolution}
Definition 2.4 and Remark 2.5 of \cite{bms} explain how a $p$-shifted linear Poisson structure on $\mathscr{E}_{g,c}[1]$ may be extended to a $p$-shifted Poisson bracket on $\mathrm{Sym}(\mathscr{E}_{g,c}[1])$. We skip over details, but mention that the Poisson bracket $\{-,-\}_{(-1)}$ defined by extending $\tau_{(-1)}$ is the usual antibracket one sees in the BV formalism corresponding to the $P_{0}$ algebra structure of \cite{cosgwill2}. The bracket $\{-,-\}_{(0)}$ defined from $\tau_{(0)}$ is the Peierls bracket, which defines a $P_{1}$ algebra structure, and only arises in special instances of local constancy in one dimension: namely, time evolution in a Lorentzian background. We continue by describing the functoriality of these constructions.
\end{rmk}

\begin{defn}[Definition 3.10 in \cite{bms}]
	A \textbf{natural collection of free BV theories} $(F_{M}, Q_{M}, (-,-)_{M}, W_{M})_{M \in \mathbf{Loc}}$ consists of natural vector bundles $\mathsf{F}^{n}$, natural linear differential operators $Q^{n} : \Gamma(\mathsf{F}^{n}) \to \Gamma(\mathsf{F}^{n+1})$ and $W^{n} : \Gamma(\mathsf{F}^{n}) \to\Gamma(\mathsf{F}^{n-1})$, and natural fiber metrics $(-,-)^{n} : \mathsf{F}^{n} \otimes \mathsf{F}^{1-n} \to \mathbf{R}$, such that for all $M \in \mathbf{Loc}$, $(F_{M}, Q_{M}, (-,-)_{M}, W_{M})$ is a free BV theory. 
\end{defn}

\begin{rmk}
	A summary of what naturality means here is in Appendix A of \cite{bms}, but the central idea is that a \textit{natural vector bundle} is a functor $\mathsf{F} : \mathbf{Loc} \to \mathbf{VectBun}$ and that a \textit{natural linear differential operator} is therefore a natural transformation $P : \Gamma(\mathsf{E}) \to \Gamma(\mathsf{F})$, whose components are thus linear differential operators $P_{M} : \Gamma(E_{M}) \to \Gamma(F_{M})$ for all $M \in \mathbf{Loc}$. 
	
	Also, we must note here that the preceding notions as well as what follows can be adapted from the site $\mathbf{Loc}$ to the subsite $\mathbf{RicLoc}$ of Ricci flat globally hyperbolic spacetimes. This follows immediately from the fact that Ricci flatness is a diffeomorphism equivariant property. 
	
\end{rmk}

Some key consequences of naturality/functoriality of free BV theories are as follows. For all $f : M \to N$ in $\mathbf{Loc}$, we have: (1) a pushforward cochain map $f_{*} : \mathscr{F}_{c}(M) \to \mathscr{F}_{c}(N)$ and a pullback cochain map $f^{*} : \mathscr{F}(N) \to \mathscr{F}(M)$, both in $\mathbf{Ch}$; (2) naturality of the integration pairing, so that 

\[\begin{tikzcd}
	{\mathscr{F}_{c}(M) \otimes \mathscr{F}(N)} &&& {\mathscr{F}_{c}(M) \otimes \mathscr{F}(M)} \\
	\\
	{\mathscr{F}_{c}(N) \otimes \mathscr{F}(N)} &&& {\mathbf{R}[-1]}
	\arrow["{\mathrm{id} \otimes f^{*}}", from=1-1, to=1-4]
	\arrow["{f_{*} \otimes \mathrm{id}}"', from=1-1, to=3-1]
	\arrow["{\langle -,- \rangle_{M}}", from=1-4, to=3-4]
	\arrow["{\langle -,- \rangle_{N}}"', from=3-1, to=3-4]
\end{tikzcd}\]
commutes; and (3) naturality of Green's witnesses: $W_{N}f_{*} = f_{*}W_{M}$ and $W_{M}f^{*} = f^{*}W_{N}$. Naturality of the remaining constructions--$P = QW + WQ$, $\Lambda_{\pm}$, etc.--follow from this. We can now reiterate a central result which follows from all of the above. 

\begin{thm}[Theorem 3.13 in \cite{bms}]\label{linearcauchyconstancy}
	Let $(F_{M}, Q_{M}, (-,-)_{M}, W_{M})_{M \in \mathbf{Loc}}$ be a natural collection of free BV theories. Then $\tau_{(0)}$ satisfies \textbf{Einstein causality}: namely, for all causally disjoint morphisms $f_{1} : M_{1} \to N \leftarrow M_{2} : f_{2}$ in $\mathbf{Loc}$, $\tau_{(0)} \circ (f_{1} \otimes f_{2}) = 0$. Moreover, for every Cauchy morphism $f : M \xrightarrow{c} N$, $f_{*} : \mathscr{F}_{c}(M)[1] \to \mathscr{F}_{c}(N) [1]$ is a quasi-isomorphism in $\mathbf{Ch}$. 
\end{thm}

Since we know how functoriality over $\mathbf{Loc}$ plays out for \textit{linear} observables, we want to consider how it plays out in the polynomial case, so we introduce the following. 

\begin{defn}\label{classicalobs}
	The \textbf{classical observables} of a free BV theory $(F,Q, (-,-), W)$ on $M$ are defined as $\mathrm{Obs^{cl}}(M) = \mathrm{Sym}(\mathscr{F}_{c}(M)[1] )$ with differential $\mathscr{Q}$ defined below.
\end{defn}

\begin{rmk}
The observables $\mathrm{Obs^{cl}}(M) = \mathrm{Sym}(\mathscr{F}_{c}(M)[1] )$ inherit the differential $Q_{[1]} = -Q$ defined on $\mathscr{F}_{c}(M)[1]$ by demanding that the Leibniz rule holds on the symmetric algebra. We will refer to this differential on $\mathrm{Obs^{cl}}(M)$ as $\mathscr{Q}$, and it satisfies $\mathscr{Q}^{2} = 0$, which is a form of the \textit{classical master equation}.  Naturality also applies in this case, and so we may view the observables as a functor $\mathrm{Obs^{cl}} : \mathbf{Loc} \to \mathbf{Ch}$. In this case, the input is $U \in \mathbf{Loc}$ and the output is $(\mathrm{Sym}(\mathscr{F}_{c}(U)[1] ), \mathscr{Q})$, which we view as an object in $\mathbf{Ch}$. We may also restrict to a given $M \in \mathbf{Loc}$ and consider the functor $\mathrm{Obs}^{\mathrm{cl}}_{M} : \mathbf{RC}_{M} \to \mathbf{Ch}$ defined analogously.
	
\end{rmk}

With this, we can make some analogous statements for those in the case of linear observables.

\begin{prop}[Adapted from Proposition 4.2 in \cite{bms}] Let $\underline{f} : \underline{M} \to N$ be a time-orderable $n$-tuple in $\mathbf{Loc}$. Then the time-ordered product 
\[\begin{tikzcd}
	{\mathrm{Obs}^{\mathrm{cl}}(\underline{M})} &&&& {\mathrm{Obs}^{\mathrm{cl}}(N)} \\
	\\
	&& {\mathrm{Obs}^{\mathrm{cl}}(N)^{\otimes n}}
	\arrow["{\mathrm{Obs}^{\mathrm{cl}}(\underline{f})}", dashed, from=1-1, to=1-5]
	\arrow["{\bigotimes_{i} f_{i *}}"', from=1-1, to=3-3]
	\arrow["{\mu^{(n)}}"', from=3-3, to=1-5]
\end{tikzcd}\]
is a cochain map: i.e., $\mathscr{Q} \circ\mathrm{Obs}^{\mathrm{cl}}(\underline{f}) = \mathrm{Obs}^{\mathrm{cl}}(\underline{f}) \circ \mathscr{Q}$. Here, $f_{i *}$ denotes the symmetric algebra extension of the pushforward cochain map $f_{i*} : \mathscr{F}_{c}(M_{i})[1] \to \mathscr{F}_{c}(N)[1]$ and $\mu^{(n)}$ denotes the $n$-ary multiplication on the symmetric algebra $\mathrm{Sym}(\mathscr{F}_{c}(N)[1]) \in \mathbf{dgCAlg}$.	
\end{prop}
The proof of this proposition follows immediately from the first line of the proof of Proposition 4.2 in \cite{bms}, since that was shown for the quantized case and the above is in the classical case.

\begin{prop}[Adapted from Proposition 4.3 in \cite{bms}]
	If $f : M \xrightarrow{c} N$ in $\mathbf{Loc}$ is a Cauchy morphism, then $\mathrm{Obs}^{\mathrm{cl}}(f) : \mathrm{Obs}^{\mathrm{cl}}(M) \to \mathrm{Obs}^{\mathrm{cl}}(N)$ is a quasi-isomorphism in $\mathbf{Ch}$. 
\end{prop}
The proof here follows from the proof of the result in the original paper; however, there is an even easier proof. Since every object in $\mathbf{Ch}$ is both fibrant and cofibrant, the free CDGA functor $\mathrm{Sym}$ preserves weak equivalences, and so this proposition follows immediately from Theorem \ref{linearcauchyconstancy}.\footnote{Thank you to Alex Schenkel for pointing this out!} 

We showed that linearized gravity $(\mathscr{E}_{g}, Q_{g}, (-,-)_{g}, W_{g})$ defines a natural collection of BV theories as long as we are careful to restrict to the subcategory $\mathbf{RicLoc}$ of Ricci flat globally hyperbolic spacetimes. This restriction is reasonable because Ricci flatness is a diffeomorphism equivariant property. By the work we showed above as well as the preceding propositions, we arrive at:

\begin{prop}\label{observablesforGR}
	The observables $\mathrm{Obs}^{\mathrm{cl}}(-, \mathscr{E})$ of linearized general relativity define a time-orderable and Cauchy constant prefactorization algebra on the site $\mathbf{RicLoc}$ of Ricci flat globally hyperbolic spacetimes.
\end{prop}

\section{Conclusion and Outlook} 

In this paper, we reviewed the theory of the site $\mathbf{Loc}$ of oriented and time-oriented globally hyperbolic spacetimes and in Proposition \ref{grothendieckequivalences} provided a helpful equivalence of coverages for that site. In Section \ref{bvgravity}, we took a look at the both perturbative general relativity and perturbative conformal gravity and were able to show how their associated formal moduli problems defined BV theories in Propositions \ref{einsteinisBV} and \ref{weylisBV}: this interpretation of what a classical field theory is aligns with the stack-oriented formalism of \cite{cosgwill2} and related work. In Section \ref{factorizationalgebras}, we truncated the $L_{\infty}$ algebra for perturbative general relativity which we studied in Section \ref{bvgravity} to its underlying cochain complex, and by shifting it down in one degree, we were able to show that linearized general relativity defines a Lorentzian BV theory in the sense of \cite{bms}. This was done by providing a Green's witness in Theorem \ref{greenswitnessGR}, which immediately implies by \cite{bms} that linearized GR is a free BV theory on the site $\mathbf{RicLoc}$ of Ricci flat, oriented, and time-oriented globally hyperbolic spacetimes. In Definition \ref{classicalobs}, we sussed out a definition of classical observables from the constructions in \cite{bms} to show in Proposition \ref{observablesforGR} that the classical observables of linearized GR define a prefactorization algebra on $\mathbf{RicLoc}$ which is particularly well-behaved with respect to causal structure.

The benefit of this work is that general relativity is now written down in detail in the formalism provided in \cite{cosgwill2}, but by showing that, we have also assembled plenty of ingredients for further study. We have also shown that the truncation to the cochain complex level satisfies the hypotheses necessary to fit the theory into the framework of the papers \cite{bps, bmsgreen, bms}, and as such we shall outline a few conjectures which we strongly expect to be true. First, we mention that although the Green hyperbolicity of the underlying cochain complex for conformal gravity $\mathscr{W}_{g}$ remains an outstanding question, we expect it to hold and that the following is thus true.

\begin{conj}
	The observables $\mathrm{Obs}^{cl}(-, \mathscr{W})$ of linearized conformal gravity define a prefactorization algebra on $\mathbf{BachLoc}$, the Bach-flat subsite of $\mathbf{Loc}$, which is time-orderable and Cauchy constant. There is an inclusion $\mathrm{Obs}^{cl}(-, \mathscr{E}) \hookrightarrow \mathrm{Obs}^{cl}(-, \mathscr{W})$ of the observables of linearized general relativity into the observables of linearized conformal gravity when restricted to $\mathbf{RicLoc}$.
\end{conj}

More importantly, there is reason to believe that the results of the preceding section hold when the theories are \textit{perturbative} and not necessarily just linear. Namely, we hope to use the $L_{\infty}$ algebras representing our field theories as in Section \ref{bvgravity} and recover the results of \cite{bms} described above, in which the authors truncated those $L_{\infty}$ algebras to their underlying cochain complexes. To do so, we will attempt to show some version of the following result.

\begin{conj}\label{LinfinityGreensWitness}
	The Green's witness $W$ for a free BV field theory $(F, Q, (-,-), W)$ which was defined by truncating an $L_{\infty}$ algebra to its underlying cochain complex may be lifted to the level of that $L_{\infty}$ algebra (perhaps only up to homotopy).
\end{conj}

Much of what we have covered so far is motivated by our study of black hole thermodynamics. In \cite{waldBH}, the author shows that any Killing horizon $\mathcal{H}$ of a stationary and globally hyperbolic black hole spacetime $M$ contains a \textit{bifurcation surface} $\Sigma$: this is the intersection of $\mathcal{H}$ with a Cauchy surface $C$ so that $\mathcal{H}$ is generated by the null geodesics orthogonal to $\Sigma$. For a spacetime $M$ of dimension $n$, Wald associates to a diffeomorphism equivariant Lagrangian $n$-form $\mathbf{L}(g)$ and a vector field $\xi$ a Noether current $(n-1)$-form $\mathbf{J}(g, \xi)$. A solution $g$ of the nonlinear Einstein equations provides a Noether charge $(n-2)$-form $\mathbf{Q}(g, \xi)$ such that $\mathbf{J} = d\mathbf{Q}$. By letting $\xi$ be a linear combination of the time and axial Killing fields for the stationary and axisymmetric black hole solution $g$, Wald defines the \textbf{black hole entropy} as
\begin{equation}\label{entropydef}
	S_{\mathrm{ent}} := \int_{\Sigma} \mathbf{Q}(g, \xi). 
\end{equation}
Moreover, if $g$ is asymptotically flat and we perturb $S_{\mathrm{ent}}$ in the direction of an asymptotically flat metric perturbation $h$, we arrive at the \textit{first law of black hole thermodynamics}:
\begin{equation}\label{firstlaw}
	\delta_{h}S_{\mathrm{ent}} = \delta_{h}\mathcal{E} - \Omega^{(\mu)}\delta_{h}\mathcal{J}_{(\mu)},
\end{equation}
where $\mathcal{E}$ and $\mathcal{J}$ are the energy and angular momentum (with angular velocity coefficients $\Omega^{(\mu)}$). It is important to note that $\delta_{h}\mathcal{E}$ and $\Omega^{(\mu)}\delta_{h}\mathcal{J}_{(\mu)}$ are quantities defined on the ``sphere at infinity",\footnote{Often referred to as the ``end" of the Cauchy surface $C$.} which we denote as $E_{\infty}$. 

Because symmetry laws and distinguished submanifolds play a central role in the above, we fully expect that factorization algebras--whose application in physics was motivated by such symmetry laws and spacetime topology--will provide an ideal home for a more mathematically nuanced formulation of black hole thermodynamics. The immediate goal is to use the results of this paper to prove a result along the lines of the following, as well as others concerning black hole entropy.

\begin{conj}
		Let $\mathscr{E}_{g}$ be the $L_{\infty}$ algebra representing the formal neighborhood of the asymptotically flat and stationary black hole solution $g \in [\mathrm{Lor}(M)/\mathrm{Diff}^{+}(M)]$, which is a BV classical field theory by Proposition \ref{einsteinisBV}, and let $\mathrm{Obs^{cl}}$ be its factorization algebra of classical observables. Then any choice of an asymptotically flat perturbation $h \in H^{0}(\mathscr{E}_{g})$ recovers a factorization algebra enhancement of Equation (\ref{firstlaw}), where the left hand side is included from $H^{0}(\mathrm{Obs^{cl}}(\Sigma))$ into $H^{0}(\mathrm{Obs^{cl}}(M))$ and the right hand side is included from $H^{0}(\mathrm{Obs^{cl}}(E_{\infty}))$.
\end{conj}

\begin{rmk}
	We are investigating a few lines of attack to prove this--the central issue being how to evaluate $\mathrm{Obs^{cl}}$ on the submanifolds $\Sigma$ and $E_{\infty}$ to include into $\mathrm{Obs^{cl}}(M)$. The usual definition of a (pre)factorization algebra allows us to compare observables over disjoint sets via the factorization product; however, multiplying observables defined over submanifolds of nontrivial codimension may be an issue--not to mention that one of the manifolds is a sphere at infinity.\footnote{This last fact suggests we may need to invoke a conformal compactification of $M$, as described in \cite{prinz}.}
	
 In this instance, the causal structure of the ambient spacetime is of  central importance. The results of \cite{bms} include the definition of a $P_{1}$ algebra structure on $\mathrm{Obs^{cl}}$, as mentioned in Remark \ref{timeevolution}, which takes advantage of local constancy in the time direction (also known as time evolution): this may allow us to project observables of the theory to codimension one Cauchy surfaces.

 We must also take advantage of asymptotic flatness. One way to enforce asymptotic flatness is cohomologically, by using an action of a DGLA: this is outlined in \cite{grady}, and will assist in further restricting the observables to the codimension two surfaces $\Sigma$ and $E_{\infty}$. To do this, it may be useful to invoke Conjecture \ref{LinfinityGreensWitness} as a lemma. 
\end{rmk}

\subsection{Acknowledgements} This work resides in an interdisciplinary confluence, so I must thank colleagues from a variety of fields. Gratitude goes to my coresearcher Ryan Grady for our continued work on an ambitious project; Shadi Tahvildar-Zadeh for helping with PDE and differential geometry details and for being a stellar postdoctoral supervisor; Philip Mannheim and Jim Wheeler for helping me understand conformal gravity; Alex Schenkel, Marco Benini, and Giorgio Musante for assisting in categorical and homological nuances in Lorentzian signature; and Owen Gwilliam for useful conversations concerning the structure of this paper. Thanks also to the anonymous referees for their insightful commentary, which greatly improved this paper. Finally, thank you to my wife Julie for being a wonderful life partner and to our little Mira for giving me a reason to work hard.

\end{document}